%% file: paper.tex
\documentclass[english]{article}
 
\makeatletter

\input{preamble.tex}

\usepackage[hidelinks]{hyperref} 
\usepackage{enumitem}
\usepackage{float}
\usepackage{setspace}
\providecommand{\tabularnewline}{\\}
\floatstyle{ruled}
\newfloat{algorithm}{tbp}{loa}
\providecommand{\algorithmname}{Algorithm}
\floatname{algorithm}{\protect\algorithmname}

\setenumerate[1]{label=\arabic*.}

\usepackage{geometry} 
\makeatletter
\newcommand{\manuallabel}[2]{\def\@currentlabel{#2}\label{#1}}
\makeatother

\begin{document} 

\title{A Fast Binary Splitting Approach to Non-Adaptive Group Testing}
\author{Eric Price and Jonathan Scarlett}

\maketitle

\input{abstract.tex}

\long\def\symbolfootnote[#1]#2{\begingroup\def\thefootnote{\fnsymbol{footnote}}\footnote[#1]{#2}\endgroup}

\symbolfootnote[0]{ E.~Price is with the Department of Computer Science, University of Texas at Austin (e-mail: \url{ecprice@cs.utexas.edu}). J.~Scarlett is with the  Department of Computer Science and the Department of Mathematics, National University of Singapore  (e-mail: \url{scarlett@comp.nus.edu.sg}).

E.~Price was supported in part by NSF Award CCF-1751040 (CAREER).  J.~Scarlett was supported by an NUS Early Career Research Award.}

\input{main_body.tex}

\appendix

\section*{\huge Appendix}

\input{appendix.tex}

\bibliographystyle{myIEEEtran}
\bibliography{JS_References}
 
\end{document}

%% file: preamble.tex
\usepackage[T1]{fontenc}
\usepackage[latin9]{inputenc}
\usepackage{amsthm}
\usepackage{amsmath}
\usepackage{amssymb}
\usepackage{graphicx}
\usepackage{dsfont}
\usepackage{cite}
\usepackage{amsmath}
\usepackage{array}
\usepackage{multirow}
\usepackage{caption}
\usepackage{color}

\usepackage{multicol}

\theoremstyle{plain}

\theoremstyle{plain}

\theoremstyle{plain}
\newtheorem{theorem}{\protect\theoremname}
\theoremstyle{plain}
\newtheorem{mainresult}{\protect\mainresultname}
\theoremstyle{plain}
\newtheorem{lemma}{\protect\lemmaname}
\theoremstyle{plain}
  
\theoremstyle{definition}

\theoremstyle{definition}

\theoremstyle{definition}
\newtheorem*{remark*}{\protect\remarkname}

\makeatother
  
\usepackage{babel} 

\providecommand{\claimname}{Claim}
\providecommand{\lemmaname}{Lemma}
\providecommand{\mainresultname}{Main Result}
\providecommand{\propositionname}{Proposition}
\providecommand{\theoremname}{Theorem}
\providecommand{\corollaryname}{Corollary} 
\providecommand{\definitionname}{Definition}
\providecommand{\assumptionname}{Assumption}
\providecommand{\remarkname}{Remark}

\newcommand{\overbar}[1]{\mkern 1.25mu\overline{\mkern-1.25mu#1\mkern-0.25mu}\mkern 0.25mu}

\newcommand{\openone}{\mathds{1}}

\newcommand{\Shat}{\widehat{S}}

\newcommand{\Ctil}{\widetilde{C}}

\newcommand{\kbar}{\overbar{k}}

\newcommand{\Dc}{\mathcal{D}}

\newcommand{\Gc}{\mathcal{G}}

\newcommand{\Tc}{\mathcal{T}}

\newcommand{\EE}{\mathbb{E}}
\newcommand{\PP}{\mathbb{P}}

\newcommand{\var}{\mathrm{Var}}
\newcommand{\cov}{\mathrm{Cov}}

\newcommand{\Tsf}{\mathsf{T}}
\newcommand{\Ssf}{\mathsf{S}}

%% file: abstract.tex
\begin{abstract}
    In this paper, we consider the problem of noiseless non-adaptive group testing under the for-each recovery guarantee, also known as probabilistic group testing.  In the case of $n$ items and $k$ defectives, we provide an algorithm attaining high-probability recovery with $O(k \log n)$ scaling in both the number of tests and runtime, improving on the best known $O(k^2 \log k \cdot \log n)$ runtime previously available for any algorithm that only uses $O(k \log n)$ tests.  Our algorithm bears resemblance to Hwang's adaptive generalized binary splitting algorithm (Hwang, 1972); we recursively work with groups of items of  geometrically vanishing sizes, while maintaining a list of ``possibly defective'' groups and circumventing the need for adaptivity.  While the most basic form of our algorithm requires $\Omega(n)$ storage, we also provide a low-storage variant based on hashing, with similar recovery guarantees.
\end{abstract}

%% file: main_body.tex
\section{Introduction} \label{sec:intro}

Group testing is a classical statistical problem with a combinatorial flavor \cite{Dor43,Du93,Ald19}, and has recently regained significant attention following new applications \cite{Ant11,Cli10,Cor05}, connections with compressive sensing \cite{Gil08,Gil07}, and most recently, utility in testing for COVID-19 \cite{Yel20,Hog20,Ver20}.  The goal is to identify a defective (or infected) subset of items (or individuals) based on a number of suitably-designed tests, with the binary test outcomes indicating whether or not the test includes at least one defective item.

In both classical studies and recent works, considerable effort has been put into the development of group testing algorithms achieving a given recovery criterion with a near-optimal number of tests \cite{Mal78,Mal80,Du93,Ati12,Ald14a,Cha14,Joh16,Sca15b,Ald15,Coj19,Coj19a}, and many solutions are known with decoding time linear or slower in the number of items.  In contrast, works seeking to further reduce the decoding time have only arisen more recently \cite{Che09,Ind10,Ngo11,Cai13,Lee15a,Ina19,Bon19a}; see Section \ref{sec:related} for an overview.

In this paper, we present a non-adaptive group testing algorithm that guarantees high-probability recovery of the defective set with $O(k \log n)$ scaling in both the number of tests and the decoding time in the case of $n$ items and $k$ defectives.  By comparison, the best previous decoding time alongside $O(k \log n)$ tests was $O(k^2 \log k \cdot \log n)$ \cite{Bon19a}, with faster algorithms incurring suboptimal $O(k \log k \cdot \log n)$ scaling in the number of tests \cite{Cai13,Lee15a}.  By a standard information-theoretic lower bound \cite[Sec.~1.4]{Ald19}, the $O(k \log n)$ scaling in the number of tests is optimal whenever $k \le n^{1-\Omega(1)}$.

Before outlining the related work and contributions in more detail, we formally introduce the problem.

\subsection{Problem Setup} \label{sec:setup}

We consider a set of $n$ items indexed by $\{1,\dotsc,n\}$, and let $S \subset \{1,\dotsc,n\}$ be the set of defective items.  The number of defectives is denoted by $k = |S|$; this is typically much smaller than $n$, so we assume that $k \le \frac{n}{2}$.  For clarity of exposition, we will present our algorithm and analysis for the case that $k$ is known, but these will trivially extend to the case that only an upper bound $\kbar \ge k$ is known, and $\kbar$ replaces $k$ in the number of tests and decoding time. 

A sequence of $t$ tests $X^{(1)},\dotsc,X^{(t)}$ is performed, with $X^{(i)} \in \{0,1\}^n$ indicating which items are in the $i$-th test, and the resulting outcomes are $Y^{(i)} = \bigvee_{j \in S} X_j^{(i)}$ (i.e., $1$ if there is any defective item in the test, $0$ otherwise).  We focus on the non-adaptive setting, in which all tests $X^{(1)},\dotsc,X^{(t)}$ must be designed prior to observing any outcomes.

We consider a for-each style recovery guarantee, in which the goal is to develop a randomized algorithm that, for any fixed defective set $S$ of cardinality $k$, produces an estimate $\Shat$ satisfying $\PP[\Shat \ne S] \le \delta$ for some small $\delta > 0$ (typically $\delta \to 0$ as $n \to \infty$).  In our algorithm, only the tests $\{X^{(i)}\}_{i=1}^t$ will be randomized, and the decoding algorithm producing $\Shat$ from the test outcomes will be deterministic.  


\subsection{Related Work} \label{sec:related}

The existing works on group testing vary according to the following defining features \cite{Du93,Ald19}:
\begin{itemize}
    \item {\bf For-each vs.~for-all guarantees.}  In {\em combinatorial (for-all) group testing}, one seeks to construct a test design that guarantees the recovery of {\em all} defective sets up to a certain size.  In contrast, in {\em probabilistic (for-each) group testing}, the test design may be randomized, and the algorithm is allowed some non-zero probability of error.  
    \item {\bf Adaptive vs.~non-adaptive.} In the {\em adaptive} setting, each test may be designed based on all previous outcomes, whereas in the {\em non-adaptive setting}, all tests must be chosen prior to observing any outcomes.  The non-adaptive setting is often preferable in practice, as it permits the tests to be performed in parallel. 
    \item {\bf Noiseless vs.~noisy.} In the {\em noiseless} setting, the test outcomes are perfectly reliable, whereas in {\em noisy settings}, some tests may be flipped according to some probabilistic or adversarial noise model.
\end{itemize}
As outlined in the previous subsection, our focus is on the for-each guarantee, non-adaptive testing, and noiseless tests, but for comparison, we also discuss other variants in this section.

\begin{table}
\begin{center}
    \begin{tabular}{|>{\centering}m{2.8cm}|>{\centering}m{2.9cm}|>{\centering}m{3.2cm}|>{\centering}m{2.6cm}|}
    \hline 
    \textbf{Reference} & \textbf{Number of tests} & \textbf{Decoding time} & \textbf{Construction}\tabularnewline
    \hline 
    \hline
    SAFFRON \cite{Lee15a}   & $O(k\cdot\log k\cdot\log n)$ & $O(k\log k)$\footnotemark & Randomized\tabularnewline
    \hline
    GROTESQUE~\cite{Cai13} & $O(k\cdot\log k\cdot\log n)$ & $O(k\cdot\log k\cdot\log n)$ & Randomized\tabularnewline
    \hline 
%
    Inan \emph{et al.}~\cite{Ina19}  & $O\big(k\cdot\log n\cdot\log\frac{\log n}{\log k}\big)$ & $\ensuremath{O\big(k^{3}\cdot\log n\cdot\log\frac{\log n}{\log k}\big)}$ & Explicit\tabularnewline
    \hline 
    BMC \cite{Bon19a} & $O(k\log n)$ & $O(k^2 \cdot\log k\cdot\log n)$ & Randomized\tabularnewline
    \hline
    {\bf This Paper} & $O(k\log n)$ & $O(k \log n)$ & Randomized\tabularnewline
    \hline
    \end{tabular}
    \par\end{center}

    \caption{Overview of existing noiseless non-adaptive group testing results with ${\rm poly}(k \log n)$ decoding time under the for-each guarantee, with $n$ items and $k$ defectives.  \label{tbl:sublinear}} \vspace*{-3ex}
\end{table}

\footnotetext{The decoding time of SAFFRON is typically stated as
  $O(k \log k \cdot \log n)$, but can be $O(k \log k)$ in the word-RAM
  model of computation, where the test outcomes are stored as
  $O(k \log k)$ words of length $\log_2 n$ each; see Appendix
  \ref{sec:saffron} for details. \label{foot:saffron}}

By a simple counting argument, even in the least stringent scenario of noiseless adaptive testing and the for-each guarantee, $\Omega\big(k \log \frac{n}{k}\big)$ tests are required for reliable recovery \cite{Mal78,Bal13}.  In the noiseless adaptive setting, the classical generalized binary splitting algorithm of Hwang \cite{Hwa72} matches this bound with sharp constant factors, and attains the stronger for-all guarantee.  More recently, $O(k \log n)$ adaptive tests and decoding time were shown to suffice under the for-each guarantee in the presence of random noise \cite{Cai13}.

In the non-adaptive setting with the for-each guarantee, a recent of result of \cite{Ald18} shows that $n$ tests are required to attain asymptotically vanishing error probability as $n \to \infty$ with $k = \lfloor \beta n \rfloor$, for arbitrarily small $\beta > 0$.  Thus, $O(k \log \frac{n}{k})$ scaling (with a universal implied constant) is impossible; see also \cite{Coj19a}.  On the other hand, $O(k \log n)$ tests do suffice, and this scaling is equivalent to $O(k \log \frac{n}{k})$ in the regime $k \le n^{1-\Omega(1)}$.  Early works
achieved this for the scaling regime $k = O(1)$
\cite{Mal78,Mal80,Ati12}, and a variety of more recent works
considered more general $k$
\cite{Mal78,Mal80,Ati12,Sca15b,Ald15,Joh16,Coj19,Coj19a}, particularly
under ``unstructured'' random test designs such as i.i.d.~Bernoulli
\cite{Ati12,Ald14a,Sca15b} and near-constant tests-per-item
\cite{Joh16,Coj19}.  In fact, following the recent work of Coja-Oghlan
{\em et al.} \cite{Coj19a}, the precise constants for non-adaptive
group testing have been nailed down in the regime
$k \le n^{1-\Omega(1)}$, with the algorithm for the upper bound
requiring $\Omega(n)$ decoding time.

\addtocounter{footnote}{-1}

Since non-adaptive algorithms with ${\rm poly}(k \log n)$ decoding time under the for-each guarantee are particularly relevant to the present paper, we provide a summary of the existing works in Table \ref{tbl:sublinear}.  We separate our discussion according to whether or not each algorithm attains $O(k \log n)$ scaling in the number of tests, as we believe this to be a particularly important desiderata in practice when tests are expensive:
\begin{itemize}
    \item SAFFRON and GROTESQUE both use $O(k\log k \cdot \log n)$ tests, which is suboptimal by a $\log k$ factor.  In particular, SAFFRON's decoding time is as low as $O(k \log k)$ in the word-RAM model.\footnotemark~While storage requirements were not a significant point of focus in the existing works, we also note (for later comparison) that SAFFRON only requires $O(k\log k \cdot \log n)$ bits of storage; see Appendix \ref{sec:saffron} for details.
    \item In the regime $k = \Theta(n^{\alpha})$ with fixed $\alpha \in (0,1)$, Inan {\em et al.}~\cite{Ina19} attain $O(k^3 \log n)$ decoding time with $O(k \log n)$ tests, whereas the number of tests becomes suboptimal in sparser regimes such as $k = {\rm poly}(\log n)$.  The best previous decoding time for an algorithm that only uses $O(k \log n)$ tests is $O(k^2 \log k \cdot \log n)$, via bit-mixing coding (BMC) \cite{Bon19a}.  To our knowledge, all other existing algorithms succeeding with $O(k \log n)$ tests incur $\Omega(n)$ decoding time \cite{Mal78,Mal80,Ati12,Sca15b,Ald15,Joh16,Coj19,Coj19a}.
\end{itemize}

Our main theoretical contribution is to provide an algorithm that attains $O(k \log n)$ scaling in both the number of tests and decoding time, as well as using distinct algorithmic ideas from the existing works (see Section \ref{sec:description} for discussion).  In particular, among the algorithms using $O(k \log n)$ tests, ours reduces the dependence on $k$ in the decoding time from quadratic to linear.


It is worth noting that the works \cite{Cai13,Lee15a,Ina19,Bon19a} also extend their guarantees to noisy settings, and the approach in \cite{Ina19} has the additional advantage of using an {\em explicit} test design, i.e., the test vectors $X^{(1)},\dotsc,X^{(t)}$ can be constructed deterministically in time polynomial in $n$ and $t$.  Although noisy and/or de-randomized variants of our algorithm may be possible, this is deferred to future work.\footnote{Our analysis can easily be adapted to handle false positive tests, but handling false negative tests appears to require more significant changes and a slightly increased decoding time. \label{foot:noise}} 

In addition, \cite{Lee15a} demonstrated that $O\big( k \log n \cdot \log\frac{1}{\epsilon} \big)$ tests suffice for SAFFRON in the case that one is only required to identify a fraction $1 - \epsilon$ of the defectives, as opposed to the entire defective set.  Thus, $O(k \log n)$ scaling is maintained for approximate recovery with constant $\epsilon > 0$, but not for the more common requirement of exact recovery.

While we have focused our discussion on the for-each setting, the first ${\rm poly}(k \log n)$-time non-adaptive group testing algorithms were for the more stringent for-all setting \cite{Che09,Ind10,Ngo11}.  
The strength of the for-all guarantee comes at the expense of requiring significantly more tests, e.g., see \cite{Dya82} for an $\Omega\big(\min\big\{k^2 \frac{\log n}{\log k},n\big\}\big)$ lower bound.  An $O(k^2 \log n)$ upper bound on the number of tests was originally attained with $\Omega(n)$ decoding time, \cite{Dya83,Por11}, and more recently with ${\rm poly}(k \log n)$ decoding time \cite{Ind10,Ngo11}.  The earlier work of \cite{Che09} attained ${\rm poly}(k \log n)$ decoding time under a list decoding recovery criterion, and also allowed for adversarial noise in the test outcomes.

\subsection{Summary of Results} \label{sec:contr}

Before summarizing our main results, we briefly highlight that our algorithmic techniques are distinct from the existing works summarized above (see Section \ref{sec:description} for further details), and are more closely related to the adaptive binary splitting approach of Hwang \cite{Hwa72}.  We first test large groups of items together, placing each group into a single randomly-chosen test among a sequence of $O(k)$ tests.  We then ``split'' these groups into smaller sub-groups, while using the $O(k)$ test outcomes to eliminate those known to be non-defective.  This process is repeated (with the elimination step ensuring that the number of groups under consideration does not grow too large) until a superset of $S$ is found.  This superset is shown to be of size $O(k)$ with high probability, and $S$ is deduced from this superset via a final sequence of $O(k \log k)$ tests (similar ideas to this final step have appeared in works such as \cite{Che09,Ngo11}).  Despite the sequential nature of this procedure, the tests can be performed non-adaptively.

Our first main result is informally stated as follows; see Theorem \ref{thm:main1} for a formal statement.

\begin{mainresult} \label{mr1}
    {\em (Algorithmic Guarantees -- Informal Version)} There exists a non-adaptive group testing algorithm that, for any constant $c > 0$, succeeds with probability $1 - O(k^{-c})$ using $O(k \log n)$ tests and $O(k \log n)$ decoding time, and requires $O(n \log^2 k)$ bits of storage.
\end{mainresult}

Main Result \ref{mr1} uses $\Omega(n)$ storage to record 
the random choices of tests the items are included in.  For our second main result, we consider a low-storage variant in which the tests are instead allocated using
hash functions.  In this case, the precise decoding time and storage
guarantees depend on the properties of the hash family used.  The
following theorem informally states the guarantee resulting from
$O(\log k)$-wise independent hashing using the classical hash family
of Wegman and Carter \cite{Weg81}; see Theorem \ref{thm:main2} for a
formal and more general statement.

\begin{mainresult} \label{mr2}
    {\em (Algorithmic Guarantees with Reduced Storage -- Informal Version)} There exists a non-adaptive group testing algorithm that, for any constant $c > 0$, succeeds with probability $1 - O\big(\frac{\log n}{k^c}\big)$ using $O(k \log n)$ tests and $O(k \log k \cdot \log n)$ decoding time, and requires $O(k \log n + \log k \cdot \log^2 n)$ bits of storage.
\end{mainresult}


%
%
\section{Non-Adaptive Binary Splitting Algorithm} \label{sec:alg_result}

In this section, we introduce our algorithm, and state and prove our main result giving guarantees on its correctness, number of tests, decoding time, and storage.  For simplicity of notation, we assume throughout the analysis that $k$ and $n$ are powers of two.  Our algorithm only requires an upper bound on the number of defectives, and hence, any other value of $k$ can simply be rounded up to a power of two.  In addition, the total number of items $n$ can be increased to a power of two by adding ``dummy'' non-defective items.

\subsection{Description of the Algorithm} \label{sec:description}

\begin{figure}
    \begin{centering}
        \includegraphics[width=0.8\columnwidth]{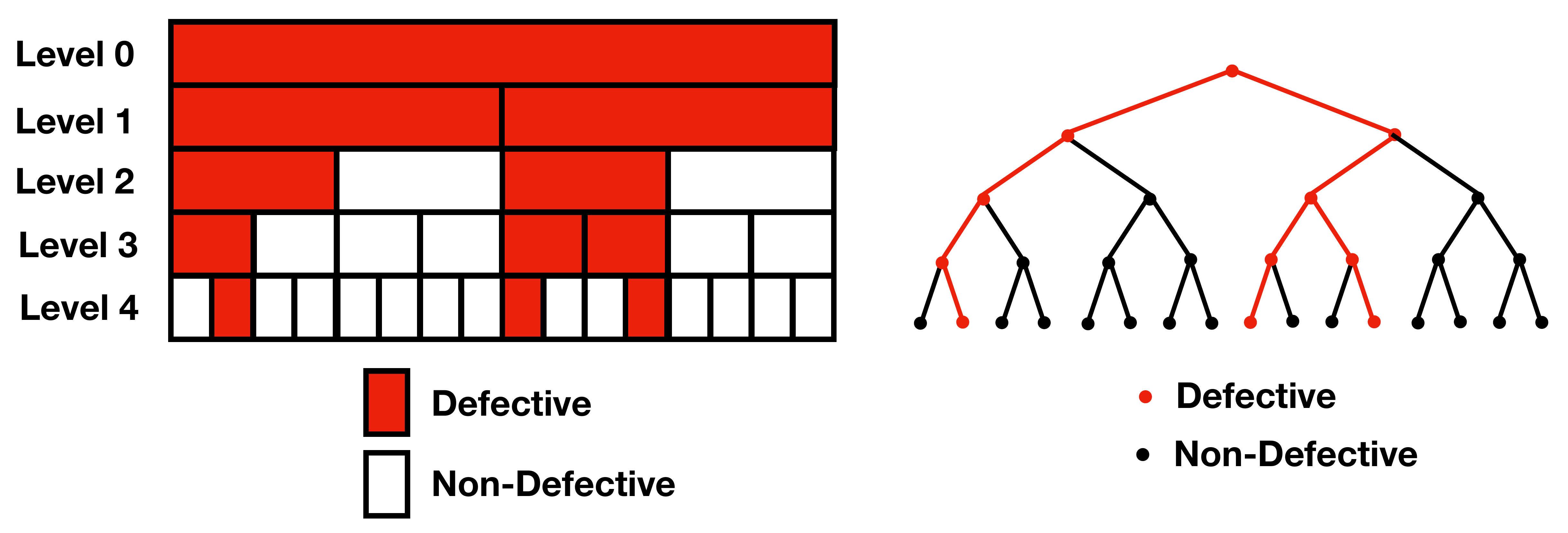}
        \par
    \end{centering}
    
    \caption{Example structure of groups of items that are tested together when $n = 16$ and $k =3$ (Left), and the corresponding tree representation (Right).  Leaves correspond to single items, and nodes at the higher levels correspond to groups of items. \label{fig:tree}} \vspace*{-1.5ex}
\end{figure}

We propose an algorithm that works with a tree structure, as illustrated in Figure \ref{fig:tree}.  On the left of the figure, we have $1+\log_2 n$ levels, with the top level containing all items, the second level containing two groups with half the items each, and so on, with each level ``splitting'' the previous levels' groups in half until the final level containing individual items.  The ordering of items is inconsequential for our analysis and results provided that the two children of a given node can be identified in constant time, so for convenience we assume the natural order.\footnote{Alternatively, one may use a dyadic splitting approach: Assign each item a unique $(\log_2 n)$-bit string, and first split according to the first bit, then the second bit, and so on.}  For $\ell = 0,\dotsc,\log_2 n$, the $j$-th group at the $\ell$-th level is denoted by $G^{(\ell)}_j \subseteq \{1,\dotsc,n\}$, and we observe that for fixed $\ell$, there are $2^{\ell}$ groups $G^{(\ell)}_1,\dotsc,G^{(\ell)}_{2^{\ell}}$ of cardinality $\frac{n}{2^{\ell}}$ each.

We consider the binary tree representation in Figure \ref{fig:tree} (Right), in which each node corresponds to a group of items (e.g., the root corresponds to the set of all items, and the leaves correspond to individual items).  The levels of the tree are indexed by $\ell = 0,\dotsc,\log_2 n$.  Our algorithm works down the tree one level at a time, keeping a list of {\em possibly defective} (PD) nodes, and performing tests to obtain such a list at the next level, while ideally keeping the size of the list small (e.g., $O(k)$).  When we perform tests at a given level, we treat each node as a ``super-item''; including a node in a test amounts to including all of the items in the corresponding group $G^{(\ell)}_j$.
While this may appear to naturally lead to an adaptive algorithm, we can in fact perform all of the tests non-adaptively.

If we were to test nodes at the root or the early levels, they would almost all return positive, and hence not convey significant information.  We therefore let the algorithm start at the level $\ell_{\min} = \log_2 k$, and we consider the following test design: 
\begin{itemize}
    \item For each $\ell = \log_2 k, \dotsc, \log_2 n - 1$, form a sequence of $C k$ tests (for some $C > 0$ to be selected later), and for each $j=1,\dotsc,2^{\ell}$, choose a single such test uniformly at random, and place all items from $G^{(\ell)}_j$ into that test.
    \item For the final level $\ell = \log_2 n$, each $G^{(\ell)}_j = \{j\}$ is a singleton.  In this case, we form $C' \log k$ sequences of $2k$ tests (for some $C' > 0$).  For each item, and each of the $C' \log k$ sequences of tests, we place the item in one of the $2k$ corresponding tests, chosen uniformly at random.
\end{itemize}
This creates a total of $t = Ck \log_2\frac{n}{k} + 2 C' k \log k = O(k \log n)$ tests.

Upon observing the $t$ non-adaptive test outcomes, the decoder forms an estimate $\Shat$ of the defective set via the following procedure:
\begin{itemize}
    \item Initialize $\Gc^{(\ell_{\min})} = \{ G^{(\ell_{\min})}_j \}_{j=1}^{k}$, where $\ell_{\min} = \log_2 k$;
    \item Iterate the following for $\ell = \log_2 k,\dotsc,\log_2 n - 1$:
    \begin{itemize}
        \item For each group $G \in \Gc^{(\ell)}$, check whether the single test corresponding to that group is positive or negative.  If positive, then add both children of $G$ (see Figure \ref{fig:tree}) to $\Gc^{(\ell+1)}$.
    \end{itemize}
    \item Let the estimate $\Shat$ of the defective set be the (singleton) elements of $\Gc^{(\log_2 n)}$ that are not included in any negative test among the $2C' k \log k$ tests at the final level.
\end{itemize}

\paragraph*{Comparisons with existing methods} Our algorithm is distinct from existing sublinear-time non-adaptive group testing algorithms, which are predominantly based on code concatenation \cite{Ind10,Ngo11} or related methods that encode item indices' binary representations directly into the test matrix \cite{Cai13,Lee15a,Bon19}.  In fact, in the context of group testing, our approach is perhaps most reminiscent of the {\em adaptive} algorithm of Hwang \cite{Hwa72} (see also \cite{Du93,Ald19}), which identifies one defective at a time using binary splitting, removing each defective from consideration after it is identified.  By keeping track of the multiple defective nodes simultaneously and allowing a small number of false positives up to the final level, we are able to exploit similar ideas without requiring adaptivity.

Beyond group testing and binary splitting, this recursive
tree-structured approach is also reminiscent of ideas used in
sketching algorithms, such as dyadic tree recovery for count-min
sketch~\cite{Cor05a}.  Our particular approach---maintaining a
superset at each level, and relying on a lack of false negatives---is
most similar to the pyramidal reconstruction method for compressive
sensing under the earth mover's distance~\cite{Ind11}.

A nearly identical algorithm to ours was given independently in the concurrent work of \cite{Che20}, with similar theoretical guarantees.  In addition, \cite{Che20} observes that since the total number of non-discarded nodes in the algorithm is $O\big(k \log \frac{n}{k}\big)$ with overwhelming probability (see Lemma \ref{lem:Ntotal_hp} below), one can take a union bound over the $n \choose k$ possible defective sets to attain \emph{list recovery} in the combinatorial (for-all) setting.  This idea allows the authors of \cite{Che20} to attain a variety of other results, including in the related problems of heavy hitters and compressive sensing.  In the context of probabilistic (for-each) group testing, some slight advantages of our analysis include (i) proving that the first batch of tests leads to $O(k)$-size list recovery rather than $O\big(k \log \frac{n}{k}\big)$-size, thus circumventing the need for a third batch of tests, and (ii) only requiring $O(\log k)$-wise independence in our low-storage variant based on hashing (see Section \ref{sec:storage}), rather than $O\big(k \log \frac{n}{k}\big)$-wise independence (albeit at the expense of a higher error probability).

\subsection{Algorithmic Guarantees} \label{sec:main_result}

In the following, we provide the formal version of Main Result \ref{mr1}.  Here and subsequently, we assume a word-RAM model of computation; for instance, with $n$ items and $t$ tests, it takes constant time to read a single integer in $\{1,\dotsc,n\}$ from memory, perform arithmetic operations on such integers, fetch a single test outcome indexed by $\{1,\dotsc,t\}$, and so on.

\begin{theorem} \label{thm:main1}
    {\em (Algorithmic Guarantees)} Let $S$ be a fixed (defective) subset of $\{1,\dotsc,n\}$ of cardinality $k$.  For any constant $c > 0$, there exist choices of $C$ and $C'$ such that the group testing algorithm described in Section \ref{sec:description} satisfies the following with probability $1 - O(k^{-c})$:
    \begin{itemize}
        \item The returned estimate $\Shat$ equals $S$;
        \item The algorithm runs in time $O(k \log n)$;
        \item The algorithm uses $O(n (\log k)^2)$ bits of storage.
    \end{itemize}
    In addition, the number of tests scales as $O(k \log n)$.
\end{theorem}
We note that the success probability approaches one as $n \to \infty$ in any asymptotic regime satisfying $k = \omega(1)$, but not in the regime $k = O(1)$; the same is true in the existing works listed in Table \ref{tbl:sublinear}, with the exception of \cite{Bon19a}.  In addition, it is straightforward to improve the success probability to $1 - e^{-\Omega(k)} - O(n^{-c})$ by using $O(\log n)$ (rather than $O(\log k)$) sequences of $2k$ tests at the final level.

Theorem \ref{thm:main1} is proved in the remainder of the section.  We emphasize that our focus is on scaling laws, and we make no significant effort to optimize the underlying constant factors.  

\subsection{Analysis} \label{sec:analysis}

Throughout the analysis, the defective set $S$ will be fixed but otherwise arbitrary, and we will condition on fixed placements of the defective items into tests (and hence, fixed test outcomes).  The test placements of the non-defective items are independent of those of the defectives, and our analysis will hold regardless of which particular tests the defectives were placed in.  We use $\PP[\cdot \,|\, \Tc_S]$ to represent this conditioning, and we refer to $\Tc_S$ as the defective test placements.  In addition, for the tree illustrated in Figure \ref{fig:tree}, we refer to nodes containing defective items as defective nodes, to all other nodes as non-defective nodes, and to the set of defective nodes as the defective (sub-)tree.\footnote{Since we start at level $\ell_{\min} = \log_2 k$, this could more precisely be considered as a forest.}

We begin with the following simple lemma.

\begin{lemma} \label{lem:marking}
    {\em (Probabilities of Non-Defectives Being in Positive Tests)} Under the test design described in Section \ref{sec:description}, the following holds at any given level $\ell = \log_2 k, \dotsc, \log_2 n - 1$: Conditioned on any defective test placements $\Tc_S$, any given non-defective node at level $\ell$ has probability at most $\frac{1}{C}$ of being placed in a positive test.
\end{lemma}
\begin{proof}
    Since there are $k$ defective items, at most $k$ nodes at a given level can be defective, and since each node is placed in a single test, at most $k$ tests out of the $Ck$ tests at the given level can be positive.   Since the test placements are independent and uniform, it follows that for any non-defective node, the probability of being in a positive test is at most $\frac{1}{C}$.  
\end{proof}
In view of this lemma, when starting at any non-defective child of any given defective node (or alternatively, starting at a non-defective node at level $\ell_{\min}$), we can view any further branches down the non-defective sub-tree as ``continuing'' (i.e., the two children remain ``possibly defective'') with probability at most $\frac{1}{C}$.  Hence, we have the following.

\begin{lemma} \label{lem:reach}
    {\em (Probability of Reaching a Non-Defective Node)} Under the setup of Lemma \ref{lem:marking}, any given non-defective node having $\Delta$ non-defective ancestor nodes is reached (i.e., all of its ancestor nodes are placed in positive tests, so the node is considered possibly defective) with probability at most $\big( \frac{1}{C} \big)^{\Delta}$.
\end{lemma}
%
%
We will use the preceding lemmas to control the following two quantities:
\begin{itemize}
    \item $N_{\rm total}$, the total number of non-defective nodes that are reached  in the sense of Lemma \ref{lem:reach};
    \item $N_{\rm leaf}$, the number of non-defective leaf nodes that are reached, and thus are still considered possibly defective by the final level.
\end{itemize}
It will be useful to upper bound $N_{\rm total}$ for the purpose of controlling the overall decoding time, and to upper bound $N_{\rm leaf}$ for the purpose of controlling the number of items considered at the final level.


\subsubsection{Bounding $N_{\rm total}$ and $N_{\rm leaf}$ on average} \label{sec:bound_avg}

We first present two lemmas bounding the averages of $N_{\rm total}$ and $N_{\rm leaf}$, and then establish high-probability bounds.

\begin{lemma} \label{lem:Ntotal_avg}
    {\em (Bounding $N_{\rm total}$ on Average)} For any $C \ge 4$ and any defective test placements $\Tc_S$, we have
    \begin{equation}
        \EE[ N_{\rm total} \,|\, \Tc_S ] \le 6 k \log_2\frac{n}{k}.
    \end{equation}
\end{lemma}
\begin{proof}
    Since the tree is binary, each defective node can have at most $2^{i}$ non-defective descendants appearing exactly $i$ levels further down the tree; such descendants correspond to $\Delta = i-1$ in Lemma \ref{lem:reach}.  Since there are at most $k \log_2\frac{n}{k}$ defective nodes, it follows that there are at most $\big( k \log_2\frac{n}{k} \big) 2^{\Delta+1}$ non-defective nodes with $\Delta$ non-defective ancestors and at least one defective ancestor.  Similarly, there are at most $k$ non-defective nodes at level $\ell_{\min} = \log_2 k$, each of which has at most $2^\Delta$ non-defective nodes appearing $\Delta$ levels later.

    Since Lemma \ref{lem:reach} demonstrates a probability $\big( \frac{1}{C} \big)^{\Delta}$ of being reached for a given number $\Delta$ of non-defective ancestors, it follows that
    \begin{equation}
        \EE[ N_{\rm total} \,|\, \Tc_S  ] \le \Big( k \log_2\frac{n}{k}\Big) \sum_{\Delta = 0}^{\log_2\frac{n}{k}-1} 2^{\Delta+1} \Big( \frac{1}{C} \Big)^{\Delta} + k\sum_{\Delta = 0}^{\log_2\frac{n}{k}} 2^{\Delta} \Big( \frac{1}{C} \Big)^{\Delta}.
    \end{equation}
    Hence, the assumption $C \ge 4$ gives
    \begin{equation}
        \EE[ N_{\rm total} \,|\, \Tc_S  ] \le \Big( 2 k \log_2\frac{n}{k}\Big) \sum_{\Delta = 0}^{\log_2\frac{n}{k}-1} 2^{- \Delta}  + k\sum_{\Delta = 0}^{\log_2\frac{n}{k}} 2^{- \Delta} \le 6 k \log_2\frac{n}{k}, \label{eq:simplifications}
    \end{equation}
    where we used $\sum_{\Delta=0}^{\infty} 2^{-\Delta} = 2$ and $\log_2 \frac{n}{k} \ge 1$ (for $k \le \frac{n}{2}$).
\end{proof}

\begin{lemma} \label{lem:Nleaf_avg}
    {\em (Bounding $N_{\rm leaf}$ on Average)}
    For any $C \ge 4$ and any defective test placements $\Tc_S$, we have
    \begin{equation}
        \EE[ N_{\rm leaf} \,|\, \Tc_S ] \le 6 k.
    \end{equation}
\end{lemma}
\begin{proof}
    Again using the fact that each defective node can have at most $2^{i}$ descendants appearing $i$ levels further down the tree, we find that there are at most $k 2^{\Delta+1}$ leaf nodes with $\Delta$ non-defective ancestors and at least one defective ancestor.  In addition, a leaf having only non-defective ancestors corresponds to $\Delta = \log_2\frac{n}{k}$ in Lemma \ref{lem:reach} (since $\ell_{\min} = \log_2 k$), and to handle this case, we trivially upper bound he number of leaves by $n$.  Hence, for $C \ge 4$, we have similarly to \eqref{eq:simplifications} that
    \begin{equation}
        \EE[ N_{\rm leaf} \,|\,\Tc_S ] \le 2 k \sum_{\Delta = 0}^{\log_2\frac{n}{k} -1} 2^{\Delta} \Big( \frac{1}{C} \Big)^{\Delta } + n\Big( \frac{1}{C} \Big)^{\log_2\frac{n}{k}} \le 6k. \label{eq:avg_leaf}
    \end{equation}
\end{proof}

\subsubsection{Bounding $N_{\rm total}$ with high probability} \label{sec:bound_Ntotal}

In this subsection, we prove the following.

\begin{lemma} \label{lem:Ntotal_hp}
    {\em (High-Probability Bound on $N_{\rm total}$)} For any $C \ge 16$, conditioned on any defective test placements $\Tc_S$, we have $N_{\rm total} = O\big( k \log\frac{n}{k} \big)$ with probability $1-e^{-\Omega(k \log \frac{n}{k})}$.
\end{lemma}
To prove this result, we make use of Lemmas \ref{lem:marking} and \ref{lem:Ntotal_avg}, along with the following auxiliary result written in generic notation.

\begin{lemma} \label{lem:branching}
    {\em (Sub-Exponential Behavior in a Branching Process)} Consider an infinite-depth binary tree in which each node is independently assigned a value $\{0,1\}$ with probability $q$ of being $1$, and let $N$ be the number of nodes whose ancestors are all marked as $1$.  Then, when $q \le \frac{1}{16}$, we have all $n > 0$ that $\PP[N = n] \le 2^{-(n-1)}$.
\end{lemma}
\begin{proof}
    We make use of branching process theory \cite[Ch.~XII]{Fel57}, in particular noting that $N$ equals the {\em total progeny} \cite{Dwa69} when the initial population size is $1$ (i.e., the root) and the distribution of the number of children per node is given by
    \begin{equation}
        P_X(0) = 1-q, \quad P_X(2) = q, \quad P_X(x) = 0,\, \forall x \notin \{0,2\}.
    \end{equation}
    The results we use from branching process theory are expressed in terms of the generative function $M_X(s) = \EE[s^X]$, which is given by
    \begin{equation}
        M_X(s) = (1-q) + qs^2. \label{eq:MX}
    \end{equation}
    It is evident in our case that $N < \infty$ with probability one provided that $q < \frac{1}{2}$, and one way to formally prove this is to utilize the general result that the {\em extinction probability} $\PP[N < \infty]$ equals one provided that the smallest root of $s = M_X(s)$ is $s=1$ \cite[Sec.~XII.4]{Fel57}.  
    
    We utilize an exact expression for the distribution of $N$ for general branching processes  \cite{Dwa69} (specialized to the case that the initial population size is one):
    \begin{equation}
        \PP[N = n] = \frac{1}{n} m_{n-1}^{(n)}, \label{eq:p_branching}
    \end{equation}
    where $m_{n-1}^{(n)}$ is computed according to the expansion
    \begin{equation}
        \big( M_X(s) \big)^n = \sum_{j=0}^{\infty} m_j^{(n)} s^j. \label{eq:Mxn}
    \end{equation}
    Substituting our expression \eqref{eq:MX} for $M_X(s)$ on the left-hand side, we find that $(M_X(s))^n = \big( (1-q) + qs^2 \big)^n$, which equals $\sum_{i=0}^n {n \choose i} (1-q)^{n-i} (qs^2)^i$ by a binomial expansion.  Hence, we obtain $m_{n-1}^{(n)} = 0$ for even-valued $n$, and for odd-valued $n$ we substitute $i = \frac{n-1}{2}$ to obtain
    \begin{align}
        m_{n-1}^{(n)} 
            &= {n \choose \frac{n-1}{2}} (1-q)^{\frac{n+1}{2}} q^{\frac{n-1}{2}} \\
            &\le 2 (2 \sqrt{q})^{n-1},  \label{eq:mn3}
    \end{align}
    where we applied ${n \choose \frac{n-1}{2}} \le 2^n = 2 \cdot 2^{n-1}$ and $(1-q)^{\frac{n+1}{2}} \le 1$.
    
    Substituting \eqref{eq:mn3} into \eqref{eq:p_branching}, we obtain 
    \begin{equation}
        \PP[N = n] \le \frac{2}{n}  \big( \sqrt{4q} \big)^{n-1} \openone\{n \text{ is odd}\}. \label{eq:pN}
    \end{equation}
    This implies $\PP[N = n] \le 2^{-(n-1)}$ when $q \le \frac{1}{16}$ and $n \ge 2$, and the same trivially holds when $n=1$.
\end{proof}

The condition $\PP[N = n] \le 2^{-(n-1)}$ in Lemma \ref{lem:branching} implies that $N$ is a sub-exponential random variable, and the same trivially follows for a branching process that only runs up to a certain depth (number of generations).
In the group testing setup of Lemma \ref{lem:Ntotal_hp}, we are adding together $O\big( k \log\frac{n}{k} \big)$ independent copies of such random variables (each corresponding to a different non-defective sub-tree) to get $N_{\rm total}$.  As a result, we can apply a standard concentration bound for sums of independent sub-exponential random variables \cite[Prop.~5.16]{Ver10} to obtain 
\begin{equation}
    \PP[ N_{\rm total} \ge \EE[N_{\rm total}\,|\,\Tc_S] + t \,|\, \Tc_S] \le e^{-\Omega(\min\{ t^2/(k\log_2\frac{n}{k}), t \})}, \label{eq:total_conc}
\end{equation}
from which Lemma \ref{lem:Ntotal_hp} follows by setting $t = \Theta\big(k \log_2\frac{n}{k}\big)$ and using $\EE[N_{\rm total}\,|\,\Tc_S] = O\big(k \log_2\frac{n}{k}\big)$ (see Lemma \ref{lem:Ntotal_avg}).  The condition $C \ge 16$ coincides with $q \le \frac{1}{16}$ in Lemma \ref{lem:branching}.

\subsubsection{Bounding $N_{\rm leaf}$ with high probability} \label{sec:bound_Nleaf}

In this subsection, we prove the following.

\begin{lemma} \label{lem:Nleaf_hp}
    {\em (High-Probability Bound on $N_{\rm leaf}$)} 
    For any $C \ge 12$, conditioned on any defective test placements $\Tc_S$, we have $N_{\rm leaf} = O(k)$ with probability $1-e^{-\Omega(k)}$.
\end{lemma}
%
%
We again use Lemma \ref{lem:marking}, along with the following auxiliary results written in generic notation.

\begin{lemma} \label{lem:leaf_tail}
    {\em (Tail Bound on Binary Tree Paths)}
    Consider a binary tree of height $h$ in which each node is independently assigned a value $\{0,1\}$ with probability $q$ of being $1$, and let $N_h$ be the number of leaf nodes that have a path of $1$'s back to the root (including both endpoints). Then $\PP[N_h \ge t] \le 4^{-(h+t)}$ for any integer $t \ge 1$.
\end{lemma}
\begin{proof}
    We use a proof by induction.  The base case is $h=1$, in which case we have $\PP[N_h > 2] = 0$, $\PP[N_h = 2] = \frac{1}{q^3}$, and $\PP[N_h \ge 1] \le \frac{2}{q^2}$ by the union bound.  These bounds satisfy the claim of the lemma due to the assumption $q \le \frac{1}{12}$.

    Now fix $h \ge 2$ and suppose that the claim is true for height $h-1$.  Let $L$ and $R$ be the number of leaf nodes reached in the left and right sub-trees of the root, and observe that
    \begin{align}
        \PP[N_h \ge t] 
            &\le \sum_{j=0}^t \PP[L \ge j \cap R \ge t-j] \\
            &= \PP[L \ge t] + \PP[R \ge t] + \sum_{j=1}^{t-1} \PP[L \ge j \cap R \ge t-j]. \label{eq:splitLR}
    \end{align}
    Applying the induction hypothesis, we find that the first two terms are at most $q 4^{-(h-1+t)} = 4q \cdot 4^{-(h+t)}$, and for $1 \le j \le t-1$, we have $\PP[L \ge j \cap R \ge t-j] \le q \cdot 4^{-(h-1+j)} \cdot 4^{-(h-1+t-j)} = 16 q \cdot 4^{-(2h+t)}$, where the multiplications by $q$ correspond to the root node being marked as $1$.  Substituting back into \eqref{eq:splitLR} gives
    \begin{equation}
        \PP[N_h \ge t] \le 8q 4^{-(h+t)} + 16 q t 4^{-(2h+t)}.
    \end{equation}
    We may assume that $t \le 2^{h}$, since for $t > 2^{h}$ we trivially have $\PP[N_h \ge t] = 0$.  Hence, the above bound simplifies to
    \begin{equation}
        \PP[N_h \ge t] \le 8q 4^{-(h+t)} + 16 q 4^{-(h+t)} 2^{-h} \le 4^{-(h+t)},
    \end{equation}
    since $h \ge 2$ implies $2^{-h} \le \frac{1}{4}$, and we have assumed $q \le \frac{1}{12}$.
\end{proof}

\begin{lemma} \label{lem:leaf_subexp}
    {\em (Sub-Exponential Behavior for Binary Tree Paths)}
    Under the setup of Lemma \ref{lem:leaf_tail}, $N_h$ satisfies $\EE[e^{\lambda N_h}] \le 1 + 4^{-h}$ for $\lambda \le \log 2$. In addition, for any integer $h_{\max} \ge 1$, if $N_1,\dotsc,N_{h_{\max}}$ are independent random variables with the same distribution as $N_h$ for the specified height, then $N = \sum_{h=1}^{h_{\max}} N_h$ satisfies $\EE[e^{\lambda N}] \le 2$ for $\lambda \le \log 2$.
\end{lemma}
\begin{proof}
    We seek to upper bound $\EE[e^{\lambda N_h}] = \sum_{t=0}^{\infty}\PP[N_h = t] e^{\lambda t}$.  Separating out the $t=0$ term and applying Lemma \ref{lem:leaf_tail}, we obtain $\EE[e^{\lambda N_h}] \le 1 + \sum_{t=1}^{\infty}4^{-(h+t)}e^{\lambda t} = 1 + 4^{-h}\sum_{t=1}^{\infty}e^{(\lambda - \log 4)t} = 1 + 4^{-h} \cdot \frac{e^{\lambda}}{4 - e^{\lambda}}$ for $\lambda \in [0,\log4)$.  In particular, if $\lambda \le \log 2$, then $\EE[e^{\lambda N_h}] \le 1 + 4^{-h}$.

    For the second part, we again consider $\lambda \le \log 2$, and use the independence assumption to write $\EE[e^{\lambda N}] = \prod_{h=1}^{h_{\max}} \EE[e^{\lambda N_h}] \le \prod_{h=1}^{h_{\max}}\big( 1 + 4^{-h}\big) \le e^{\sum_{h=1}^{h_{\max}} 4^{-h}  } \le e^{1/3} \le 2$.
\end{proof}
We now consider the following decomposition of $N_{\rm leaf}$:
\begin{equation}
    N_{\rm leaf} = N'_{\rm leaf} + N''_{\rm leaf}, \label{eq:N_decomp}
\end{equation}
where $N'_{\rm leaf}$ counts the reached non-defective leaf nodes having at least one defective ancestor, and $N''_{\rm leaf}$ counts the reached non-defective leaf nodes having all non-defective ancestors.  
In the case of a single defective item (i.e., $k=1$), we claim that $N'_{\rm leaf}$ has the same distribution as $2N$, with $N$ given in Lemma \ref{lem:leaf_subexp} for a suitably-chosen value of $h_{\max}$ and $q = \frac{1}{C}$.  To see this, we identify the leaf nodes in Lemma \ref{lem:leaf_subexp} with nodes at the {\em second-last} level of the tree illustrated in Figure \ref{fig:tree}, and observe that the factor of $2$ arises since every such node produces two children at the final level (hence contributing to $N'_{\rm leaf}$) when placed in a positive test.

In the case of $k$ defective items, some care is needed, as the relevant sets of leaves associated with the $k$ defective paths may overlap, potentially creating complicated dependencies between the associated random variables.  However, if we take the definition of $N$ in Lemma \ref{lem:leaf_subexp} and remove some leaves from the count, the quantity $\EE[e^{\lambda N}]$ can only decrease further, so the conclusion $\EE[e^{\lambda N}] \le 2$ remains valid.
Hence, by counting any overlapping leaves only once (in an otherwise arbitrary manner), we can form $k$ independent random variables $N^{(1)},\dotsc,N^{(k)}$ satisfying $2\sum_{i=1}^k N^{(i)} \stackrel{\rm d}{=} N'_{\rm leaf}$ and $\EE[e^{\lambda N^{(i)}}] \le 2$.  In addition, since there are $k$ nodes at level $\ell_{\min} = \log_2 k$, we can similarly form $k' \le k$ independent random variables $N^{(k+1)},\dotsc,N^{(k+k')}$ satisfying $\sum_{i=k+1}^{k+k'} N^{(i)} \stackrel{\rm d}{=} N''_{\rm leaf}$ and $\EE[e^{\lambda N^{(i)}}] \le 2$. 

Thus, by \eqref{eq:N_decomp}, $N_{\rm leaf}$ is the sum of at most $2k$ independent sub-exponential random variables.  Again applying a standard concentration bound \cite[Prop.~5.16]{Ver10}, it follows that
\begin{equation}
    \PP[ N_{\rm leaf} \ge \EE[N_{\rm leaf}\,|\,\Tc_S] + t \,|\, \Tc_S] \le e^{-\Omega(\min\{ t^2/k, t \})}, \label{eq:leaf_conc}
\end{equation}
from which Lemma \ref{lem:Nleaf_hp} follows upon setting $t = \Theta(k)$ and using the fact that $\EE[N_{\rm leaf}\,|\,\Tc_S] = O(k)$ (see Lemma \ref{lem:Nleaf_avg}).

\subsubsection{Analysis of the final level} \label{sec:final_level}

Recall that at the final level, we perform $C' \log k$ independent sequences of $2k$ tests, with each item being randomly placed in one test in each sequence.  We study the error probability conditioned on the high-probability event that $N_{\rm leaf} = O(k)$ (see Lemma \ref{lem:Nleaf_hp})

For a given non-defective item and a given sequence of $2k$ tests, the probability of colliding with any defective item is at most $\frac{1}{2}$, similarly to Lemma \ref{lem:marking}.  Due to the $C' \log k$ repetitions, for any fixed $c' > 0$, there exists a choice of $C'$ yielding $O(k^{-c'})$ probability of a given non-defective appearing only in positive tests.  By a union bound over the $N_{\rm leaf} = O(k)$ non-defectives at the final level, we find that the estimate $\Shat$ equals $S$ with (conditional) probability $1 - O(k^{1-c'})$.

\subsubsection{Number of tests, error probability, decoding time, and storage} \label{sec:numbers}

The claims of Theorem \ref{thm:main1} are established as follows:
\begin{itemize}
    \item {\em Number of tests:} As stated in Section \ref{sec:description}, the number of tests is $t = Ck \log_2\frac{n}{k} + 2 C' k \log k$, which behaves as $O\big( k \log \frac{n}{k} + k \log k \big) = O(k \log n)$.
    \item {\em Error probability}: The concentration bounds on $N_{\rm leaf}$ and $N_{\rm total}$ (see Lemmas \ref{lem:Ntotal_hp} and \ref{lem:Nleaf_hp}) hold with probability $1-e^{-\Omega(k)}$ and $1-e^{-\Omega(k\log\frac{n}{k})}$ respectively, and we treat their complements as error events.  These terms are dominated by the final stage, which incurs $O(k^{1-c'})$ failure probability; setting $c' = c+1$ gives the $O(k^{-c})$ behavior stated in Theorem \ref{thm:main1}.
    \item {\em Decoding time}: We claim that conditioned on the the high-probability events $N_{\rm total} = O\big( k \log \frac{n}{k} \big)$ and $N_{\rm leaf} = O(k)$, the decoding time is $O\big(k \log \frac{n}{k} + k \log k\big) = O(k \log n)$.  The first term comes from considering all levels except the last, since it takes $O(1)$ time to check whether each node associated with $N_{\rm total}$ is in a positive or negative test.  While $N_{\rm total}$ only counts the non-defective nodes, the number of defective nodes also trivially behaves as $O\big(k \log \frac{n}{k}\big)$.  At the last level, for each of the $k + N_{\rm leaf} = O(k)$ relevant leaf nodes, we perform $O(\log k)$ such checks for a total time of $O(k \log k)$. 
    \item {\em Storage}: At the $\ell$-th level, we need to store $2^{\ell}$ integers indicating the test associated with each of the $2^{\ell}$ nodes.  Hence, excluding the last level, we need to store $\sum_{\ell=\log_2 k}^{\log_2 n - 1} 2^{\ell} = O(n)$ integers in $\{1,\dotsc,2k\}$, or $O(n \log k)$ bits.  Similarly, the $O(\log k)$ independent sequences of tests at the final level amount to storing $O(n \log k)$ integers, or $O(n (\log k)^2 )$ bits.  In addition, under the high-probability event $N_{\rm total} = O\big(k \log\frac{n}{k} \big) = O(k \log n)$, the storage of the possibly defective set requires $O(k \log n)$ integers, or $O(k \log^2 n)$ bits.  Since $k \le n$, this is no higher than $O(n \log^2 k)$.
\end{itemize}

%
%

\section{Storage Reductions via Hashing} \label{sec:storage}

A notable weakness of Theorem \ref{thm:main1} is that the storage required at the decoder is higher than linear in the number of items.  In this section, we present a variant of our algorithm with considerably lower storage that attains similar guarantees to Theorem \ref{thm:main1}.

The idea is to interpret the mappings $\big\{1,\dotsc,2^{\ell}\big\} \to \{1,\dotsc,Ck\}$ at each level as hash functions. Since the high storage in Theorem \ref{thm:main1} comes from explicitly storing the corresponding $2^{\ell} = O(n)$ values, the key to reducing the overall storage is to use lower-storage hash families.  The reduced storage comes at the expense of reduced independence between different hash values (e.g., only $r$-wise independence for some $r \ll n$), and the proof of Theorem \ref{thm:main1} utilizes full independence.  The modified algorithm in this section uses $\Theta(\log k)$-wise independent hash families, and we leave open the question of whether similar recovery guarantees can be attained with only $O(1)$-wise independence.

A second modification to the algorithm is that in order to attain the behavior $O(k^{-c})$ in the error probability similarly to Theorem \ref{thm:main1}, we consider the use of $\Ctil \ge 1$ independent repetitions at each level $\ell = \ell_{\min},\dotsc,\log_2 n -1$, in the same way that we already used repetitions at the final level.  This means that at each level, there are $\Ctil$ sequences of $Ck$ tests (each with a different and independent hash function), and a node is only considered to be possibly defective at a given level if {\em all} of its associated $\Ctil$ tests are positive.

\begin{remark*}
    This idea of using low-storage hash functions is reminiscent of the Bloom filter data structure \cite{Blo70}.  A direct approach to transferring Bloom filters to group testing is to perform $L = O(\log n)$ hashes of each item into one of $t = O(k \log n)$ tests, and then estimate the defective set to be the set of items only included in positive tests \cite[Sec.~1.7]{Ald19}.  By checking the $L$ test outcomes associated with each item, the defective set can be identified with low storage (depending on the hash family properties) under the preceding scaling laws.  However, a drawback of this direct approach is that checking all items separately takes $\Omega(n)$ time.  Our algorithm circumvents this via the binary splitting approach.
\end{remark*}


\subsection{Statement of Result} \label{sec:main_result2}

%

With the above-described modifications to the algorithm, we have the following counterpart to Theorem \ref{thm:main1}, which formalizes and generalizes Main Result \ref{mr2}.

\begin{theorem} \label{thm:main2}
    {\em (Algorithmic Guarantees with Reduced Storage)} Let $S$ be a fixed (defective) subset of $\{1,\dotsc,n\}$ of cardinality $k$.  For any constant $c > 0$, there exist choices of $C$, $\Ctil$, and $C'$ such that the above-described group testing algorithm adapted from Section \ref{sec:description}, with a $\Theta(\log k)$-wise independent hash family and $\Ctil$ independent repetitions per level, yields the following with probability at least $1 - O\big( \frac{\log n}{k^c} \big)$:
        \begin{itemize}
            \item The returned estimate $\Shat$ equals $S$;
            \item The algorithm runs in time $O(\Tsf_{\rm hash} k \log n)$, where $\Tsf_{\rm hash}$ is the evaluation time for one hash value;
            \item The algorithm uses $O(k \log n + \Ssf_{\rm hash} \log n)$ bits of storage, where $\Ssf_{\rm hash}$ is the number of bits of storage required for one hash function.
        \end{itemize}
        In addition, the number of tests scales as $O(k \log n)$.
\end{theorem}
We briefly discuss some explicit values that can be attained for $\Tsf_{\rm hash}$ and $\Ssf_{\rm hash}$.  Supposing that $n$, $k$, and $C$ are powers of two, we can adopt the classical approach of Wegman and Carter \cite{Weg81} and consider a random polynomial over the finite field ${\rm GF}(2^{m})$, where $m \in \{ \log_2 k, \dotsc, \log_2 n\}$ (depending on the level).   In this case, one attains $r$-wise independence while storing $r$ elements of ${\rm GF}(2^{m})$ (or $O(r \log n)$ bits), and performing $O(r)$ additions and multiplications in ${\rm GF}(2^{m})$ to evaluate the hash.  As a result, with $r = \Theta(\log k)$, we get
\begin{equation}
    \Tsf_{\rm hash} = O(\log k), \qquad \Ssf_{\rm hash} = O(\log k \cdot \log n)
\end{equation}
under the assumption that operations in ${\rm GF}(2^{m})$ can be performed in constant time.  Hence, Theorem \ref{thm:main2} gives $O(k \log k \cdot \log n)$ decoding time and $O(k \log n + \log k \cdot \log^2 n)$ storage.  Different trade-offs can also be attained using more recent hash families that can attain $r$-wise independence with an evaluation time significantly less than $r$ \cite{Sie89,Tho17}.

The error probability of $O\big( \frac{\log n}{k^c} \big)$ is slightly worse than the $O(k^{-c})$ scaling of Theorem \ref{thm:main1}; in particular, $o(1)$ error probability is only guaranteed when $k = (\log n)^{\Omega(1)}$.  We expect that this requirement could be avoided via a refined analysis, e.g., by characterizing the behavior of the random variables $N_{\rm total}$ and $N_{\rm leaf}$ used in the proof of Theorem \ref{thm:main1}.  In addition, the following variants of Theorem \ref{thm:main2} can already be attained with almost no additional effort:
\begin{itemize}
    \item We can allow $c = \omega(1)$ in the statement of Theorem \ref{thm:main2}, but at the expense of the number of tests and decoding time increasing by a multiplicative $\Theta(c)$ factor.
    \item It is straightforward to show that $\EE[N_{\rm total}] = O\big( k \log \frac{n}{k} \big)$ and $\EE[N_{\rm leaf}] = O(k)$, and one can use Markov's inequality to deduce that $N_{\rm total} = O\big(k^{1+c}\log \frac{n}{k}\big)$ and $N_{\rm leaf} = O(k^{1+c})$ with probability $1-O(k^{-c})$.  For any constant $c > 0$, this approach leads to $O(k^{-c})$ error probability with $O(k \log n)$ tests, but the decoding time increases to $O(\Tsf_{\rm hash} k^{1+c} \log n)$, and the storage increases to $O(k^{1+c} \log^2 n + \Ssf_{\rm hash} \log n)$.
\end{itemize}
In the remainder of the section, we provide the proof of Theorem \ref{thm:main2}.

\subsection{Analysis} \label{sce:analysis2}

\subsubsection{Auxiliary variance calculation (case $\Ctil = 1$)} \label{sec:aux_var}

We first study the case $\Ctil = 1$ (i.e., no repetitions), as the case $\Ctil > 1$ will then follow easily.

Recall that the algorithm maintains an estimate of the possibly defective (PD) set at each level.  We will give conditions under which the size of this set remains at $O(k)$ throughout the course of the algorithm.  For $\ell = \ell_{\min} = \log_2 k$, we trivially have at most $k \le 4k$ PD items.  We will use an induction argument to show that every level has at most $4k$ PD items, with high probability.

Consider two non-defective nodes indexed by $u,v$ at a given level $\ell$ having $k' \le k$ defective nodes, let $\Dc_u,\Dc_v$ denote the respective events of hashing into a test containing one or more defectives, and let $D_u,D_v$ be the corresponding indicator random variables.  The dependence of these quantities on $\ell$ is left implicit.  We condition on all of the test placements performed at the earlier levels, accordingly writing $\EE_{\ell}[\cdot]$ and $\var_{\ell}[\cdot]$ for the conditional expectation and conditional variance.  In accordance with the above induction idea, we assume that there are at most $4k$ PD nodes at level $\ell$.

\begin{lemma} \label{lem:var_bound}
    {\em (Mean and Variance Bounds)} Under the preceding setup and definitions, if there are at most $4k$ PD nodes at level $\ell$, then we have the following when $\Ctil = 1$ and $C \ge 8$:
    \begin{gather}
        \EE_{\ell}\bigg[ \sum_{u} D_u \bigg] \le \frac{k}{2} \label{eq:E_sum} \\
        \var_{\ell}\bigg[ \sum_{u} D_u \bigg] = c_{\rm var} k, \label{eq:var_sum} 
    \end{gather}
    where the sums are over all non-defective PD nodes at the $\ell$-th level, and $c_{\rm var} > 0$ is a universal constant.
\end{lemma}
The proof is given in Appendix \ref{app:var}.  Given this result, we easily deduce the following.

\begin{lemma} \label{lem:induction}
    For $\Ctil = 1$ and $C \ge 8$, conditioned on the $\ell$-th level having at most $4k$ possibly defective (PD) nodes, the same is true at the $(\ell+1)$-th level with probability $1 - O\big(\frac{1}{k}\big)$.
\end{lemma}
\begin{proof}
    Among the PD nodes at the $\ell$-th level, at most $k$ are defective, amounting to at most $2k$ children at the next level.  By Lemma \ref{lem:var_bound} and Chebyshev's inequality, with probability $1 - O\big(\frac{1}{k}\big)$, at most $k$ non-defective nodes are marked as PD, thus also amount to at most $2k$ additional children at the next level, for a total of $4k$.
\end{proof}

%

\subsubsection{Analysis of the error probability} \label{sec:pe}

At the first level $\ell = \log_2 k$, we trivially have $k \le 4k$ possibly defective (PD) nodes.  For $\Ctil = 1$, using Lemma \ref{lem:induction} and an induction argument, the same follows for all levels simultaneously with probability at least $1 - O(k^{-1} \log n)$.  For $\Ctil > 1$, we note that since we have $\Ctil$ repetitions at each level and only keep the nodes whose tests are {\em all} positive, the $1 - O(k^{-1})$ behavior becomes $1 - O(k^{-\Ctil})$ due to the independence of the repetitions.  Hence, the expression $1 - O(k^{-1} \log n)$ for $\Ctil = 1$ generalizes to $1 - O\big(k^{-\Ctil} \log n\big)$.

The analysis of the final level in Section \ref{sec:final_level} did not rely on $h(\cdot)$ being a fully independent hash function, but rather, only relied on a collision probability of $\frac{1}{2k}$ between any two given items.  Since this condition still holds for any pairwise (or higher) independent hash family, we immediately deduce the same conclusion: Conditioned on the final level having $O(k)$ nodes marked as possibly defective, and by choosing $C'$ appropriately in the algorithm description, we attain $O(k^{1-c'})$ error probability at this level for any fixed $c' > 0$.  

Combining the above and setting $\Ctil = c$ and $c' = 1+c$, we attain the desired scaling $O(k^{-c} \log n)$ in the theorem statement.

\subsubsection{Number of tests, decoding time, and storage} \label{sec:numbers2}

The remaining claims of Theorem \ref{thm:main2} are established as follows:
\begin{itemize}
    \item {\em Number of tests:} The number of tests is the same as in the fully independent case, possibly with a modified implied constant if $\Ctil > 1$ and/or a different choice of $C$ is used.
    \item {\em Decoding time}: The analysis of the decoding time is similar to the fully independent case (see Section \ref{sec:numbers}), but each hash takes $\Tsf_{\rm hash}$ time to compute.  Hence, the decoding time is $O\big(\Tsf_{\rm hash} k \log n\big)$.
    \item {\em Storage}: We use $O(1)$ hash functions at each level except the last, and $O( \log k )$ hash functions at the final level, for a total of $O\big( \log\frac{n}{k} + \log k \big) = O(\log n)$, requiring $O(\Ssf_{\rm hash} \log n)$ storage.  In addition, under the high-probability event that there are $O(k)$ possibly defective nodes at each level, their storage requires $O(k)$ integers, or $O(k \log n)$ bits.  Hence, the total storage is $O(k \log n + \Ssf_{\rm hash} \log n)$.
\end{itemize}

\section{Conclusion} \label{sec:conclusion}

We have presented a novel non-adaptive group testing algorithm ensuring high-probability (for-each) recovery with $O(k \log n)$ scaling in both the number of tests and decoding time.  In addition, we presented a low-storage variant with similar guarantees depending on the hash family used.  An immediate open question for this variant is whether similar guarantees hold for $O(1)$-wise independent hash families.  In addition, even for the fully independent version, it would be of significant interest to develop a variant that is robust to random noise in the test outcomes (see Footnote \ref{foot:noise} on Page \pageref{foot:noise}).


%% file: appendix.tex
\section{Discussion on the SAFFRON Algorithm} \label{sec:saffron}

While the SAFFRON algorithm \cite{Lee15a} is based on sparse-graph codes, we find it most instructive to compare against the simplified singleton-only version \cite[Sec.~5.4]{Ald19}, as the more sophisticated version does not attain better scaling laws (though it may attain better constant factors).

Singleton-only SAFFRON is briefly outlined as follows.  One forms $O(k \log k)$ ``bundles'' of tests of size $2 \log_2 n$ each, and assigns each item to any given bundle with probability $O\big(\frac{1}{k}\big)$.  If a given item is assigned to a given bundle, then its $(\log_2 n)$-bit description is encoded into the first $\log_2 n$ tests (i.e., the item is included in the test if and only if its bit description contains a $1$ at the corresponding location), and its bit-wise complement is encoded into the last $\log_2 n$ tests.  The following decoding procedure ensures the identification of any defective item for which there exists a bundle in which it is included without any other defective items:
\begin{itemize}
    \item For each bundle, check whether the first $\log_2 n$ test outcomes equal the bit-wise negation of the second $\log_2 n$ outcomes.
    \item If so, add the item with bit representation given by the first $\log_2 n$ outcomes into the defective set estimate.
\end{itemize}
Due to the first step, the second step will never erroneously be performed on a bundle without defective items, nor on a bundle with multiple defective items.

In \cite{Lee15a}, a decoding time of $O(k \log k \cdot \log n)$ is stated under the assumption that reading the $2 \log_2 n$ bits takes $O(\log n)$ time.  However, if the associated $2 \log_2 n$ tests are stored in memory as two ``words'' of length $\log_2 n$ each, then a word-RAM model of computation only incurs $O(1)$ time per bundle, or $O(k \log k)$ time overall.

We also note that SAFFRON only requires $O(k \log k \cdot \log n)$ bits of storage, amounting to negligible storage overhead beyond the test outcomes themselves.  This is because the only item indices stored are those added to the estimate of the defective set, and there are at most $k$ such indices (or $k \log_2 n$ bits) due to the fact that SAFFRON makes no false positives. 

\section{Proof of Lemma \ref{lem:var_bound} (Mean and Variance Bounds)} \label{app:var}

For ease of notation, we leave the subscripts $(\cdot)_{\ell}$ implicit throughout the proof, but the associated conditioning is understood to apply to all probabilities, expectations, variance terms, and so on.

    We first prove \eqref{eq:E_sum}.  The event $\Dc_u$ occurs if $u$ is hashed into the same bin as any of the $k' \le k$ defective nodes.  Since we are hashing into $\{1,\dotsc,Ck\}$ and the hash family is (at least) pairwise independent, each collision occurs with probability $\frac{1}{Ck}$.  Hence, by the union bound, $u$ is in a positive test with probability at most $\frac{1}{C}$, and \eqref{eq:E_sum} follows from the assumption that there are at most $4k$ PD nodes and $C \ge 8$. 

    As for \eqref{eq:var_sum}, we first characterize $\cov[D_u,D_v]$, writing
    \begin{align}
        \cov[D_u,D_v] 
            &= \EE[D_u D_v] - \EE[D_u] \EE[D_v] \\
            &= \PP[\Dc_u \cap \Dc_v] - \PP[\Dc_u]\PP[\Dc_v] \\
            &= \PP[\Dc_u] + \PP[\Dc_v] - \PP[\Dc_u \cup \Dc_v] - \PP[\Dc_u]\PP[\Dc_v]. \label{eq:p_diff}
    \end{align}
    We proceed by bounding $\PP[\Dc_u]$ (the same bound holds for $\PP[\Dc_v]$) and $\PP[\Dc_u \cup \Dc_v]$ separately.
    
    \paragraph*{Probability of the individual event}  Fix a non-defective node $u$.  Let $h(\cdot)$ denote the random hash function with output values in $\{1,\dotsc,Ck\}$, and for each defective node indexed by $i \in \{1,\dotsc,k'\}$ (where $k'$ is the total number of defective nodes at the level under consideration), let $B_i$ be the ``bad'' event that $h(i) = h(u)$.  We apply the inclusion-exclusion principle, which is written in terms of the following quantities for $j = 1,\dotsc,k'$:
    \begin{equation}
         T_j = \sum_{1 < i_1 < \dotsc < i_j < k'} \PP[B_{i_1} \cap \dotsc \cap B_{i_j}].
    \end{equation}
    If the hash function is $(j+1)$-wise independent, this simplifies to 
    \begin{equation}
         T_j = {k' \choose j} \Big(\frac{1}{Ck}\Big)^j.
    \end{equation}
    Hence, if the hash function is $(j_{\max}+1)$-wise independent for some $j_{\max}$, then the inclusion-exclusion principle gives
    \begin{align}
        \PP[\Dc_u] 
            &= \PP\bigg[\bigcup_{i=1,\dotsc,k'} B_i \bigg] \\
            &\le \sum_{j=1}^{j_{\max}} (-1)^{j+1} T_j \label{eq:inc_exc0} \\
            &= \sum_{j=1}^{j_{\max}} {k' \choose j} (-1)^{j+1} \Big(\frac{1}{Ck}\Big)^j \label{eq:inc_exc}
    \end{align}
    for odd-valued $j_{\max}$, and the reverse inequality for even-valued $j_{\max}$.  Using the fact that $\sum_{j=1}^{k'} {k' \choose j} (-1)^{j+1} \big(\frac{1}{Ck}\big)^j = 1 - \big(1-\frac{1}{Ck}\big)^{k'}$, we can write \eqref{eq:inc_exc} as
    \begin{equation}
        \PP[\Dc_u] \le 1 - \Big(1-\frac{1}{Ck}\Big)^{k'} - \sum_{j_{\max}+1}^{k'} {k' \choose j} (-1)^{j+1} \Big(\frac{1}{Ck}\Big)^j. \label{eq:inc_exc2}
    \end{equation}
    The final term can then be bounded as follows in absolute value:
    \begin{align}
        \bigg| \sum_{j_{\max}+1}^{k'} {k' \choose j} (-1)^{j+1} \Big(\frac{1}{Ck}\Big)^j \bigg|
            &\le \sum_{j_{\max}+1}^{k'} {k' \choose j} \Big(\frac{1}{Ck}\Big)^j \\
            &\le \sum_{j_{\max}+1}^{\infty} \Big(\frac{1}{C}\Big)^j \label{eq:abs_bound2} \\
            &= \frac{(1/C)^{j_{\max}+1}}{1-1/C}, \label{eq:abs_bound3} 
    \end{align}
    where \eqref{eq:abs_bound2} uses ${k' \choose j} \le (k')^j$ and $k' \le k$, and \eqref{eq:abs_bound3} applies the geometric series formula.  Assuming $C \ge 2$, we can further upper bound the above expression by $(1/C)^{j_{\max}}$, and hence by any target value $\delta_0$ provided that $j_{\max} \ge \log_{C}\frac{1}{\delta_0}$.  Recall also that \eqref{eq:inc_exc2} is reversed for even-valued $j_{\max}$, so loosening the preceding requirement to $j_{\max} \ge \lceil \log_{C}\frac{1}{\delta_0} \rceil + 1$ gives
    \begin{equation}
        1 - \Big(1-\frac{1}{Ck}\Big)^{k'} - \delta_0 \le \PP[\Dc_u] \le 1 - \Big(1-\frac{1}{Ck}\Big)^{k'} + \delta_0. \label{eq:d0_deviation}
    \end{equation}
    Since $u$ is arbitrary, the same bound also holds for $\PP[\Dc_v]$.
    
    \paragraph*{Probability of the union of two events} We can decompose $\PP[\Dc_u \cup \Dc_v]$ as follows:
    \begin{align}
        \PP[\Dc_u \cup \Dc_v] 
            &= \Big(1 - \frac{1}{Ck}\Big) \PP[\Dc_u \cup \Dc_v \,|\, h(u) \ne h(v)] + \frac{1}{Ck}  \PP[\Dc_u \cup \Dc_v \,|\, h(u) = h(v)] \\
            &= \PP[\Dc_u \cup \Dc_v \,|\, h(u) \ne h(v)] + O\Big( \frac{1}{k} \Big).
    \end{align}
    Hence, we focus on the case $h(u) \ne h(v)$ in the following.  Similar to the above, let $B'_i$ be the ``bad'' event that $h(i) \in \{ h(u), h(v) \}$, and define
    \begin{equation}
         T'_j = \sum_{1 < i_1 < \dotsc < i_j < k'} \PP[B'_{i_1} \cap \dotsc \cap B'_{i_j} \,|\, h(u) \ne h(v)].
    \end{equation}
    If the hash function is $(j+2)$-wise independent, this simplifies to 
    \begin{equation}
         T'_j = {k' \choose j} \Big(\frac{2}{Ck}\Big)^j,
    \end{equation}
    where the factor of two comes from the possibility of colliding with either $u$ or $v$.  Following the same argument as above, we find that if (i) $j_{\max} \ge \lceil \log_{C/2}\frac{1}{\delta_0} \rceil + 1$, (ii) $C \ge 4$, and (iii) the hash function is $(j_{\max} + 2)$-wise independent, then the following analog of \eqref{eq:d0_deviation} holds:
    \begin{equation}
        1 - \Big(1-\frac{2}{Ck}\Big)^{k'} - \delta_0 \le \PP[\Dc_u \cup \Dc_v \,|\, h(u) \ne h(v)] \le 1 - \Big(1-\frac{2}{Ck}\Big)^{k'} + \delta_0. \label{eq:d0_deviation2}
    \end{equation}
    
    \paragraph*{Combining and simplifying} Setting $\delta_0 = \frac{1}{k}$ and combining the above findings, we deduce that for a $\Theta(\log k)$-wise independent hash,
    \begin{gather}
        \PP[\Dc_u] = 1 - \Big(1-\frac{1}{Ck}\Big)^{k'} + O\Big( \frac{1}{k} \Big), \\
        \PP[\Dc_u \cup \Dc_v] = 1 - \Big(1-\frac{2}{Ck}\Big)^{k'} + O\Big( \frac{1}{k} \Big).
    \end{gather}
    The idea in the following is to approximate $1 - \frac{\nu}{Ck} \approx e^{-\frac{1}{Ck}}$ for $\nu=1,2$, and substitute into \eqref{eq:p_diff}.  To make this more precise, we use the fact that $k' \le k$ to write
    \begin{align}
        \Big(1 - \frac{1}{Ck}\Big)^{k'} 
            &= \bigg( e^{-\frac{1}{Ck} + O\big( \frac{1}{k^2} \big)} \bigg)^{k'} \\
            &= e^{-\frac{k'}{Ck} + O\big( \frac{1}{k} \big)} \\
            &= e^{-\frac{k'}{Ck}} + O\Big( \frac{1}{k} \Big).
    \end{align}
    Applying a similar argument to  $\big(1 - \frac{2}{Ck}\big)^{k'}$ and substituting into \eqref{eq:p_diff}, we obtain
    \begin{align}
        \cov[D_u,D_v] 
            &= 2\Big(1 - e^{-\frac{k'}{Ck}}\Big) - \Big(1 - e^{-\frac{2k'}{Ck}}\Big) - \Big(1 - e^{-\frac{k'}{Ck}}\Big)^2 + O\Big( \frac{1}{k} \Big) \\
            &= O\Big( \frac{1}{k} \Big),
    \end{align}
    since the first three terms cancel upon expanding the square.  The proof is concluded by writing
    \begin{align}
        \var\bigg[ \sum_{u} D_u \bigg] = \sum_{u} \var[D_u] + \sum_{u \ne v} \cov[D_u, D_v] = O(k)
    \end{align}
    since $\var[D_u] \le \EE[D_u] \le \frac{1}{C}$, and there are at most $4 k$ values of $u$ by assumption.


%
